\documentclass{article}
\usepackage{color}
\usepackage[utf8]{inputenc}
\usepackage[T1]{fontenc}
\usepackage{cite}
\usepackage{fixltx2e}
\usepackage{graphicx}
\usepackage{url}
\usepackage{multirow}
\usepackage{amsmath}
\usepackage[standard]{ntheorem}
\usepackage{subfigure}
\usepackage[english]{babel}
\usepackage[T1]{fontenc}
\usepackage{textcomp}
\usepackage{lmodern}
\usepackage[utf8]{inputenc}
\usepackage{pxfonts}
\usepackage{dsfont}
\usepackage{stmaryrd}
\usepackage{listings}
\usepackage{todonotes}
\usepackage{figlatex}

\begin{document}

\title{Protein structure prediction software generate two different
  sets of conformations\\Or the study of unfolded self-avoiding walks}
\author{Jacques M. Bahi, Christophe Guyeux, Jean-Marc Nicod,\\ and
  Laurent Philippe\thanks{Authors in alphabetic order}\\
  Computer science laboratory DISC,\\ FEMTO-ST Institute, UMR 6174 CNRS\\
 University of Franche-Comt\'{e}, Besan\c con, France\\
 \textit{\{jacques.bahi, christophe.guyeux, jean-marc.nicod, laurent.philippe\}@femto-st.fr}}
\maketitle

\begin{abstract}
  Self-avoiding walks (SAW) are the source of very difficult problems
  in probabilities and enumerative combinatorics. They are also of
  great interest as they are, for instance, the basis of protein structure prediction
  in bioinformatics. Authors of this article have previously shown that,
  depending on the prediction algorithm, the sets
  of obtained conformations differ: all the self-avoiding walks can be reached
  using stretching-based algorithms whereas only the folded SAWs can be attained
  with methods that iteratively fold the straight line.
  A first study of (un)folded self-avoiding walks is
  presented in this article. The contribution is majorly a survey
  of what is currently known about these sets. In particular we provide 
  clear definitions of various subsets of self-avoiding walks
  related to pivot moves (folded or unfoldable SAWs, etc.) and the first results
  we have obtained, theoretically or computationally, on these sets.
  A list of open questions is provided too, and the consequences on the
  protein structure prediction problem is finally investigated.
\end{abstract}

\section{Introduction}

Self-avoiding walks (SAW) have been studied over decades, both for
their interest in mathematics and their applications in physics:
standard model of long chain polymers~\cite{Flory49}, fundamental
example in the theory of critical phenomena in equilibrium statistical
mechanics~\cite{Gordon11,Gennes72}, and so on.  They are the source of
very difficult problems in probabilities and enumerative
combinatorics~\cite{Bousquet11,Beaton12}, regarding among other things the number of $n-$step
SAW, their mean-square displacement, and the so-called scaling limit.
The self-avoiding walks naturally appear in bioinformatics, during the
prediction of the 3D conformation of a protein of interest.
Frequently, the two dimensional backbone of the protein is looked for
in a first stage, and then this 2D structure is refined step by step
to obtain the final 3D conformation.

Protein Structure Prediction (PSP) software can be separated into two
categories.  On the one hand, some algorithms construct the proteins'
structures on the 2D or 3D square lattice by adding, at each
iteration, a new amino acid at the queue of the protein.  Most of the
time, various positions are possible for this amino acid, and the
chosen position is the one that optimizes a given functional (for
instance, the number of neighboring hydrophobic amino acids).  On the
other hand, some algorithms start from the straight line having the
size of the considered protein, and they iterate pivot moves on this
structure, pivot amino acids and angles being chosen to optimize
another time a well-defined functional.  We have pointed out, in our
previous researches on the dynamics of the protein folding
process~\cite{bgcs11:ij,bgc11:ip}, that these two categories of
protein structure prediction software cannot produce the same
conformations~\cite{guyeux:hal-00795127}.  More precisely, in the
first category, all the conformations can be attained whereas it is
not the case in the second one.

Indeed this result, which is ignored by bioinformaticians, has been
formerly discovered by the community of mathematicians that studies
the self-avoiding walks (SAWs), even though the connection with
the PSP problem has not been signaled. In their article introducing the pivot
algorithm~\cite{Sokal88}, Madras and Sokal have demonstrated a theorem
showing that, when starting from the straight line of length $n$, and
iterating the 180° rotation and either both 90° rotations or both
diagonal reflections, all the $n-$step self-avoiding walks on
$\mathds{Z}^2$ can be obtained (or, in other words, their pivot
algorithm is ergodic for this set of transformations).  As a
counterexample, they depicted in this article a 223-step SAW in
$\mathds{Z}^2$ that is not connected to any other SAW by 90° rotations
(their counterexample is represented in Figure~\ref{SokalnrSAW}).  This first apparition of an
``unfolded'' SAW was indeed the unique one in the literature, and the study of (un)folded SAWs 
has not been deepened before our work in~\cite{guyeux:hal-00795127}.

In this article, the authors' intention is to produce a list of first
results and questionings about various sets of self-avoiding walks that
can (or cannot) be attained by $\pm 90$° pivot moves, and to deduce
consequences regarding the PSP software. After having recalled some basis
on self-avoiding walks, we provide definitions of 4 subsets of SAWs that
appear when considering such pivot moves, namely the folded SAWs obtained by 
iterating pivot moves on the straight line, the unfoldable
SAWs, the set of SAWs that can be folded at least once, and finally the subset
of self-avoiding walks that can be folded $k$ times, $k>1$. Then a list of
results we have obtained on these subsets is provided. Among other things, 
the Cardinality of folded SAWs  has been bounded, the infinite number of unfoldable
SAWs is established (the proof of this result, too long to be presented in this
article, is given in~\cite{articleTheoreme}), a shorter example of unfoldable walk
is given (107 steps), whereas the equality between the set of SAWs and the set of 
folded SAWs has been computationally verified until $n\leqslant 14$.
Relation between these subsets is then provided, before listing various open
problems on (un)foldable self-avoiding walks. Theoretical aspects of this
study are deepened in~\cite{articleTheoreme} whereas the computational ones
are detailed in~\cite{articleCalculs}.

The remainder of this document is organized as follows.  In the next
section, a short overview about the self-avoiding walks is
provided. This section enables us to introduce basic definitions and
well-known results concerning these walks.  Section~\ref{sec:first results}
contains the rigorous definition of the subsets of self-avoiding walks 
regarded in this manuscript.
Then, in
Section~\ref{sec:automatic gen}, the first results we have obtained
concerning the subset of unfolded SAW are detailed,
whereas a non-exhaustive list of open
questions is drawn up in Section~\ref{sec:openquestions}. Consequences
regarding the protein structure prediction problem are investigated in
Section~\ref{sec:consequences}. This research work ends by a
conclusion section, in which the contributions are summarized and
intended future work is proposed.

\section{A Short Overview of Self-Avoiding Walks}

We firstly recall usual notations and well-known results regarding
self-avoiding walks. We will bring partially, in a next section, 
these results in the folded SAWs subset.

\subsection{Definitions and Terminologies}

Let $\mathds{N}$ be the set of all natural numbers,
$\mathds{N}^*=\{1,2,\hdots\}$ the set of all positive integers, and
for $a,b\in \mathds{N}$, $a<b$, the notation $\llbracket
a,b\rrbracket$ stands for the set $\{a, a+1, \hdots, b-1, b\}$.  $|x|$
stands for the Euclidean norm of any vector $x\in \mathds{Z}^d,
d\geqslant 1$, whereas $x_1, \hdots, x_d$ are the $d$ coordinates of
$x$.  The $n-$th term of a sequence $s$ is denoted by $s(n)$.
Finally, $\sharp X$ is the Cardinality of a finite set $X$.

Let us now introduce the notion of self-avoiding walk~\cite{Madras93,Gordon11,Hughes95}.

\begin{figure}
\centering
\includegraphics[scale=0.5]{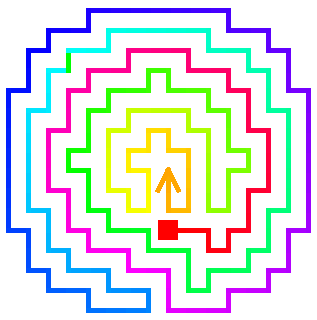}
\caption{The first SAW shown to be not connected
to any other SAW by 90° rotations (Madras and Sokal,~\cite{Sokal88}), that is, the first discovered unfoldable SAW.}
\label{SokalnrSAW}
\end{figure}

\begin{definition}[Self-Avoiding Walk]
  Let $d\geqslant 1$. A $n-$step \emph{self-avoiding walk} from
  $x\in\mathds{Z}^d$ to $y\in\mathds{Z}^d$ is a map $w:\llbracket
  0,n\rrbracket \rightarrow \mathds{Z}^d$ with:
\begin{itemize}
\item $w(0)=x$ and $w(n)=y$,
\item $|w(i+1)-w(i)|=1$, 
\item $\forall i,j \in \llbracket 0,n\rrbracket$, $i\neq j \Rightarrow
  w(i)\neq w(j)$ (self-avoiding property).
\end{itemize}
\end{definition}

Let $d\in \mathds{N}^*$.  $\mathcal{S}_n(x)$ is the set of $n-$step
self-avoiding walks on $\mathds{Z}^d$ from 0 to $x$, $c_n(x)=\sharp
\mathcal{S}_n(x)$ is the Cardinality of this set,
$\mathcal{S}_n=\cup_{x\in\mathds{Z}^d} \mathcal{S}_n(x)$ is
constituted by all $n-$step self-avoiding walks that start from 0,
whereas $c_n=\sum_{x\in\mathds{Z}^d} c_n(x)$ is the number of $n-$step
self-avoiding walks on $\mathds{Z}^d$ starting from 0, that is, $c_n =
\sharp\mathcal{S}_n$~\cite{Gordon11}.

\subsection{Well-known results about self-avoiding walks}

The objective of this section is not to realize a complete state of the art
about established or conjectured results on SAWs, but only to present
a few list of properties that are connected to our first investigations
regarding the folded self-avoiding walks. For instance, the well-known pattern 
theorem~\cite{Madras93} is not presented here. For further results about
SAWs, readers can consult for instance~\cite{Gordon11, Madras93}.

A first result concerning the number of $n-$step self-avoiding walks can be easily
obtained by remarking that, when $m-$step SAWs are concatenated to
$n-$step SAWs, we found all $(m+n)-$step self-avoiding walks \emph{and}
other walks having intersections. In other words,

\begin{proposition}
  $\forall m,n \in \mathds{N}^*, c_{m+n}\leqslant c_m c_n$.
\end{proposition}

The existence of the so-called \emph{connective constant} is a consequence of such 
a proposition.

\begin{theorem}
  The limit $\lim_{n \rightarrow \infty} c_n^{1/n}$ exists. It is
  called the \emph{connective constant} and is denoted by
  $\mu$. Moreover, we have $\mu^n \leqslant c_n$ and $d \leqslant \mu
  \leqslant 2d-1$.
\end{theorem}

Various bounds or estimates can be found in the
literature~\cite{jensen04,Gordon11}, like $c_n \approx A \mu^n
n^{\gamma -1}$ for $A$ and $\gamma$ to determine (predicted asymptotic
behavior) and
$$\mu \in [2.625662,2.679193].$$

The pivot algorithm is a dynamic Monte Carlo algorithm that produces
self-avoiding walks using the following basic approach~\cite{Sokal88}. 
Firstly, a point $p$ on the walk $w$ is picked randomly and used as a pivot. 
Then a random symmetry operation of the lattice, like a rotation,
 is applied to the second part (suffixes) of the walk, using $p$ as origin.
If the resulting walk is a SAW, it is accepted, else
it is rejected and $w$ is counted once again in the sample. 
A more detailed and precise algorithm can be found in~\cite{Sokal88}. 
In this article, it is shown that, quoting Madras and Sokal,

\begin{theorem}
  The pivot algorithm is ergodic for self-avoiding walks on
  $\mathds{Z}^d$ provided that all axis reflections, and either all
  90° rotations or all diagonal reflections, are given nonzero
  probability. In fact, any $N-$step SAW can be transformed into a
  straight rod by some sequence of $2N-1$ or fewer such pivots.
\end{theorem}

The pivot algorithm is ergodic too for SAWs on the square lattice~\cite{Sokal88}, 
provided that the 180° rotation, and either both 90°
rotations or both diagonal reflections, are given nonzero probability,
whereas 90° rotations alone are not enough, due to Fig.~\ref{SokalnrSAW}.

\section{Introducing the (un)folded self-avoiding walks}
\label{sec:first results}

\subsection{Protein folding as preliminaries}

Let us introduce the original context motivating the study of particular
subsets of SAWs we called ``folded'' self-avoiding walks in the remainder
of this document.


In the 2 or 3 dimensional square lattice hydrophobic-hydrophilic model,
simply denoted as \emph{HP model}, which is used for low resolution 
backbone structure prediction
of a given protein, hydrophobic interactions are supposed to dominate
protein folding~\cite{bgcs11:ij,bgc11:ip}.
 This model was formerly introduced by Dill~\cite{Dill1985}, who
considers that the protein core freeing up energy is
formed by hydrophobic amino acids, whereas hydrophilic amino acids
tend to move in the outer surface due to their affinity with the
solvent (see Fig.~\ref{fig:hpmodel}).

In this model, a protein conformation is a SAW on a 2D or 3D lattice,
depending on the level of resolution. This SAW is such that the free
energy $E$ of the protein, which depends on topological neighboring
contacts between hydrophobic amino acids that are not contiguous in
the primary structure, is minimal.  In other words, for an amino acid
sequence $P$ of length $n$ and for the set $\mathcal{C}(P)$ of all
$n-$step SAWs, the walk chosen to represent the conformation of the
protein is $C^* = min \left\{E(C) \mid C \in
  \mathcal{C}(P)\right\}$ \cite{Shmygelska05}.  In that context and
for a conformation (SAW) $C$, $E(C)=-q$ where $q$ is equal to the
number of topological hydrophobic neighbors.  For example, $E(c)=-5$
in Fig.~\ref{fig:hpmodel}.

The overriding problem in PSP is: \emph{how to find such a minimal conformation, given
all the  $n-$step self-avoiding walks and the sequence of hydrophobicity
of the protein~?}

\begin{figure}[t]
\centering
\includegraphics[width=2.375in]{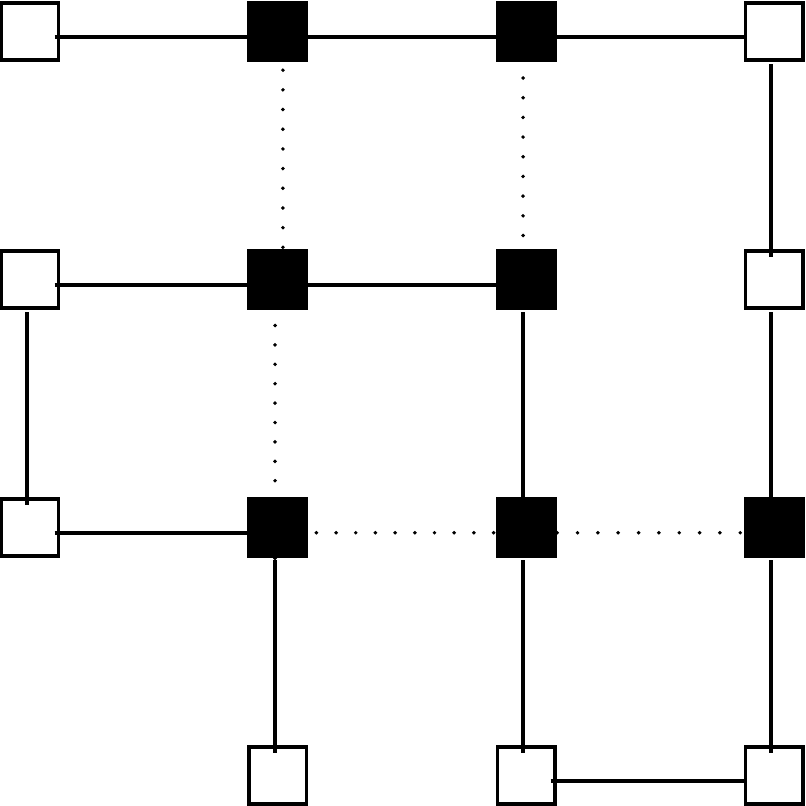}
\caption{Hydrophilic-hydrophobic model (black squares are
hydrophobic residues)}
\label{fig:hpmodel}
\end{figure}

To find the best 2D conformation of a protein, given its sequence of
hydrophobicity, is really not an easy task. Indeed authors
of~\cite{Crescenzi98} have proven that, considering the set of
self-avoiding walks having $n-$steps and whose vertices are either
black (hydrophobic) or white squares (hydrophylic residues), to
determine the SAWs of this set that maximize the number of neighboring
black squares is NP-hard.  Given a sequence of amino acids, such
statement leads to the use of heuristics to predict (and not to
determine exactly) the most probable conformation of the
protein. These heuristics operate as in the real biological world,
folding or increasing the length of SAWs in order to minimize the free
energy of the associated conformation: by doing so, the protein
synthesis in aqueous environment is reproduced \emph{in silico}.  As
stated previously, we have shown in a previous work that such
investigations potentially lead to various subsets of self-avoiding
walks~\cite{bgcs11:ij,bgc11:ip,guyeux:hal-00795127}.

In the first approach, starting from the straight line, we obtain by a
succession of pivot moves of 90° a final conformation being a
self-avoiding walk.  In this approach, it is not regarded whether the
intermediate walks are self-avoiding or not.  Such a method
corresponds to programs that start from the initial conformation, fold
several times the linear protein, according to their embedded scoring
functions, and then obtain a final conformation on which the SAW
requirement is verified.  It is easy to be convinced that, by doing
so, the set of final conformations is exactly equal to the set of
self-avoiding walks having $n$ steps. As the conformations obtained by
such methods coincide exactly to the well-studied global set of all
SAWs, such an approach is not further investigated in what
follows~\cite{guyeux:hal-00795127}.

\begin{figure}
\centering
\includegraphics[scale=0.7]{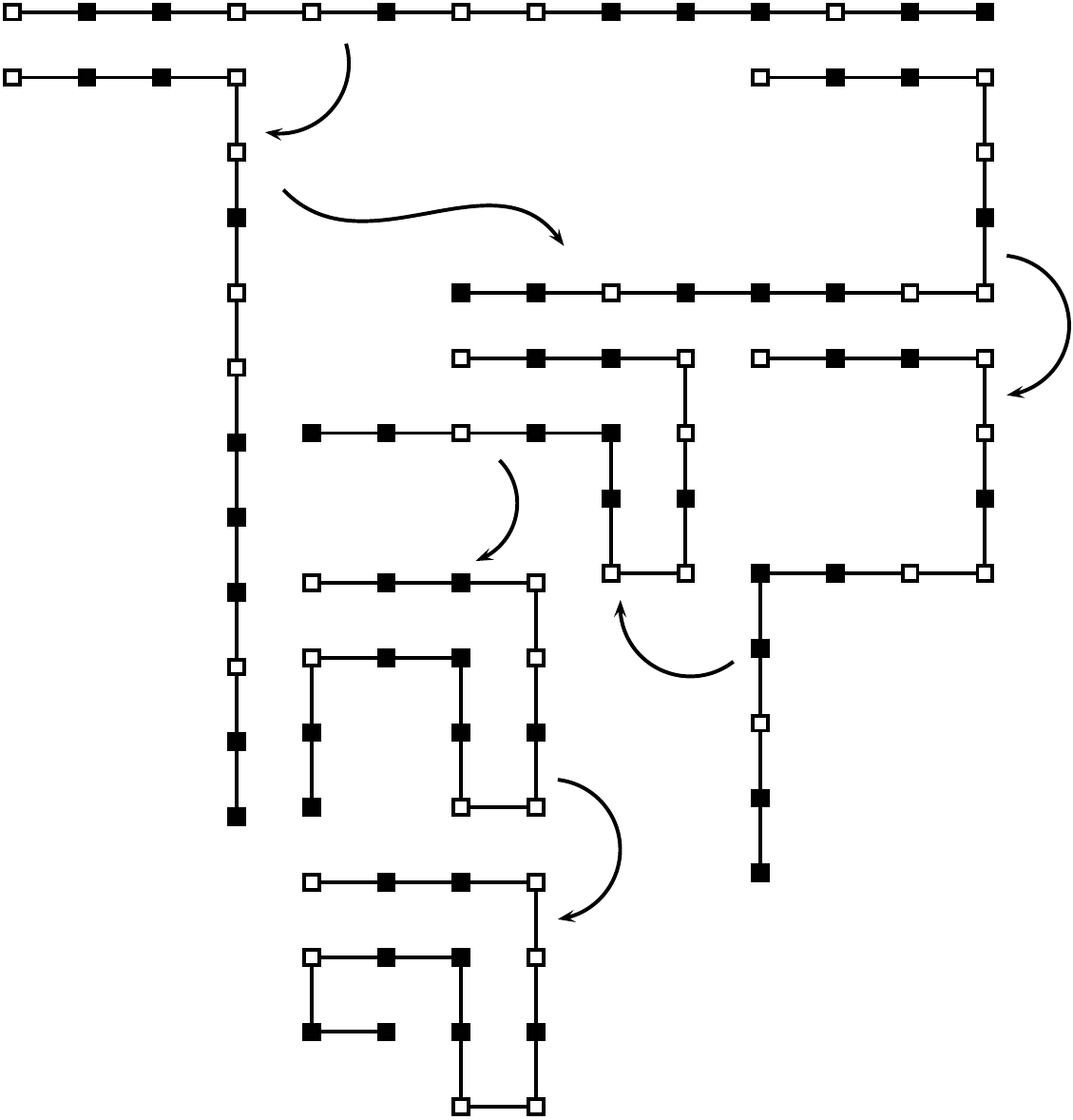}
\caption{Protein Structure Prediction by folding SAWs}
\label{exPliage}
\end{figure}

In the second approach, the same process is realized, except that all
the intermediate conformations must be self-avoiding walks (see
Fig.~\ref{exPliage}).  The set of $n-$step SAWs reachable by such a
procedure is denoted by $fSAW_n$ in what follows.  Such a procedure is
one of the two most usual translations of the so-called ``SAW
requirement'' in the bioinformatics literature, leading to proteins'
conformations belonging into $fSAW_n$.  For instance, PSP methodes
presented in~\cite{DBLP:conf/cec/IslamC10, Unger93,Braxenthaler97,
  DBLP:conf/cec/HiggsSHS10, DBLP:conf/cec/HorvathC10} follow such an
approach.  We have shown in~\cite{guyeux:hal-00795127} that
$fSAW_n\subsetneq \mathcal{S}_n$~\cite{Sokal88}. In other words,
\emph{in this first category of PSP software, it is impossible to
  reach all the conformations of $\mathcal{S}_n$}.

Other approaches in the same category can be imagined, like the
following one.  We can act as above, requiring additionally that no
intersection of vertex or edge during the transformation of one SAW to
another occurs. For instance, the pivot move of Figure~\ref{saw4} is
authorized in the previous $fSAW$ approach, but it is refused in the
current one: during the rotation around the residue having a cross,
the rigid structure after this residue intersects the remainder of
the ``protein'' (see Fig.~\ref{saw4bis}).  In this two dimensional
approach denoted by $fSAW'$, it is impossible for a protein folding
from one plane conformation to another plane one to use the 3D space
to achieve this folding.  A reasonable modeling of the true natural
folding dynamics of an already synthesized protein can be obtained by
extending this requirement to the third dimension.  However, due to
its complexity, this requirement is actually never used by tools that
embed a 2D HP square lattice model for protein structure prediction.
This is why these particular SAWs are not really investigated in
this document.  Let us just emphasize that $fSAW_n'$ is obviously a
subset of $fSAW_n$, but there is \emph{a priori} no reason to consider
them equal.  Indeed, Figure~\ref{saw3pasSaw4} shows that,

\begin{proposition}
For all $n\in\mathds{N}^*$, $fSAW_n'\subset fSAW_n$. However, 
$\exists n\in\mathds{N}^*$, $fSAW_n'\neq fSAW_n$.
\end{proposition}

\begin{proof}
In Figure~\ref{saw3pasSaw4}, the unique possible
pivot move is the red dot, and obviously such move leads to the intersection
between the head and the queue of the structure during the transformation.
\end{proof}

\begin{figure}
\centering
\includegraphics[scale=0.4]{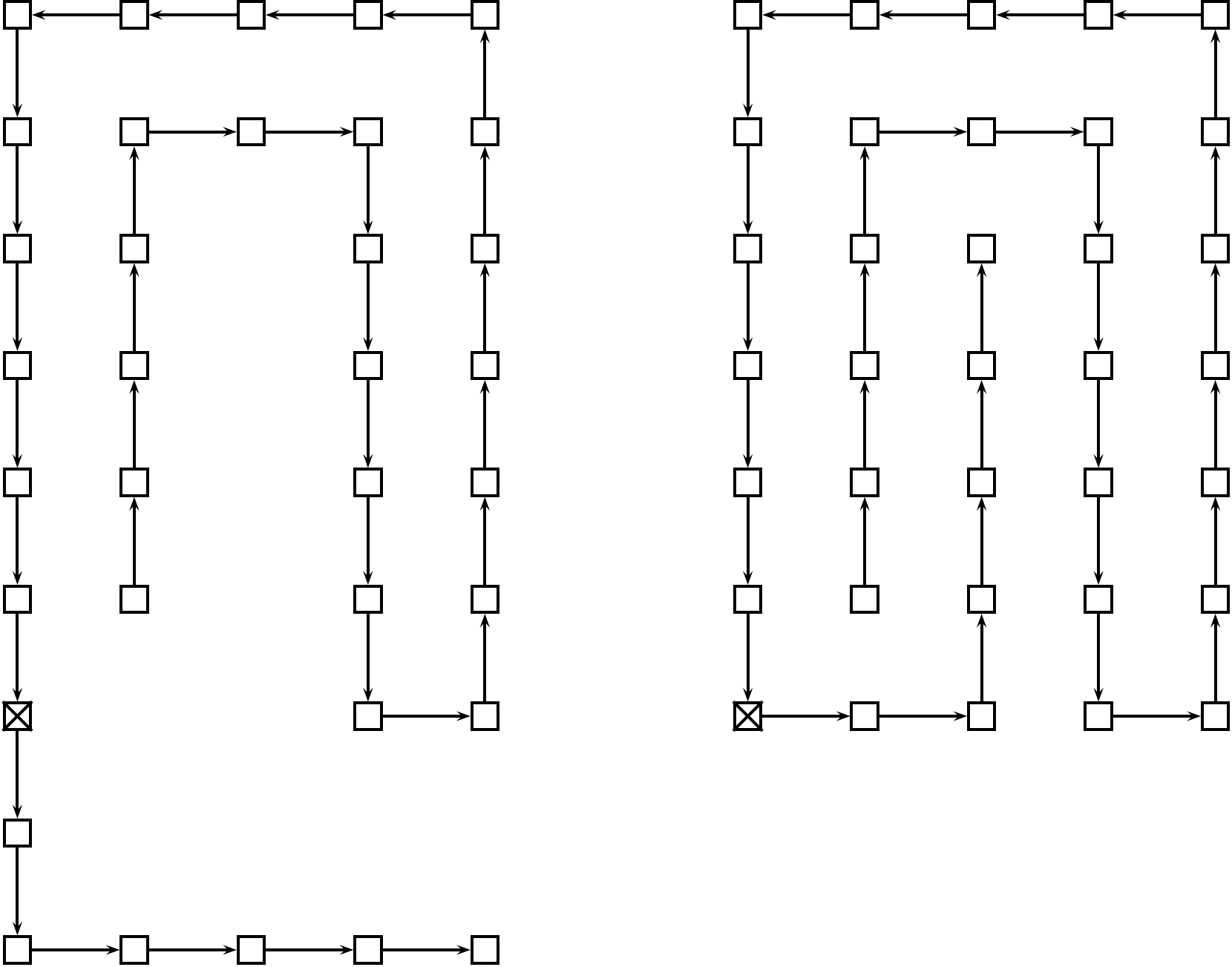}
\caption{Pivot move acceptable in $fSAW$ but not in $fSAW'$}
\label{saw4}
\end{figure}

\begin{figure}
\centering
\includegraphics[scale=0.35]{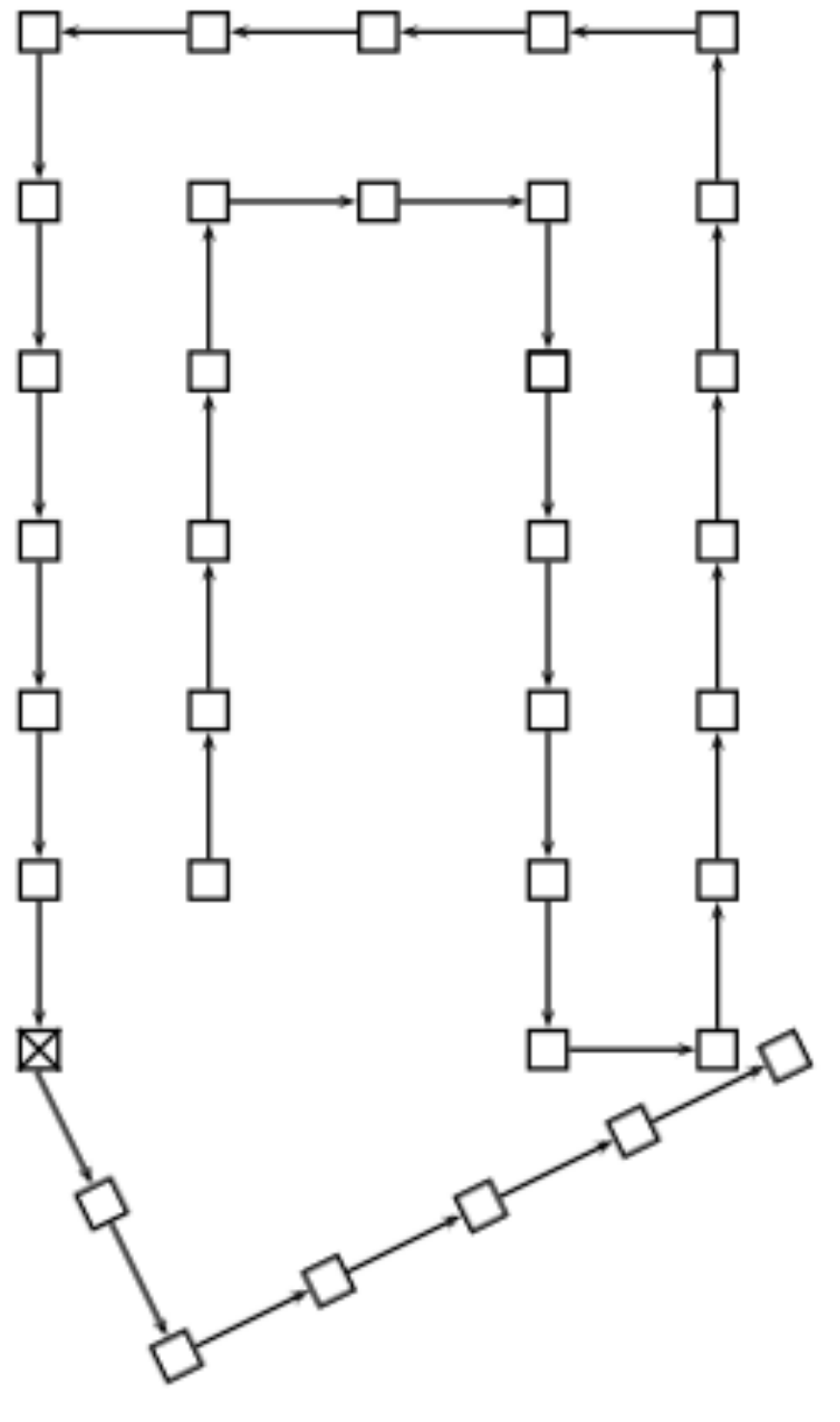}
\caption{An intersection appears between the head and the queue during
  the transformation, thus this pivot move is refused in $fSAW'$.}
\label{saw4bis}
\end{figure}

\begin{figure}
\centering
\includegraphics[scale=0.5]{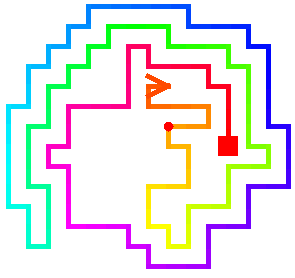}
\caption{$fSAW_n \neq fSAW_n'$}
\label{saw3pasSaw4}
\end{figure}

Note that we only studied pivot moves of $\pm 90$° in the three
previous approaches. But to consider other sets of transformations
could be interesting in some well-defined contexts, which can
potentially lead to different new subsets of SAWs.

A last bioinformatics approach of protein structure prediction using
self-avoiding walks starts with an $1-$step SAW, and at iteration $k$, a
new step is added at the queue of the walk, in such a way that the new
$k-$step self-avoiding walk presents the best value for the considered
scoring function (see Fig~\ref{exElongation}).  The protein is thus
constructed step by step, reaching the best local conformation at each
iteration.  It is easy to see that such an approach leads, another
time, to all the possible self-avoiding walks having the length of the
considered protein~\cite{guyeux:hal-00795127}.

\begin{figure}
\centering
\includegraphics[scale=0.7]{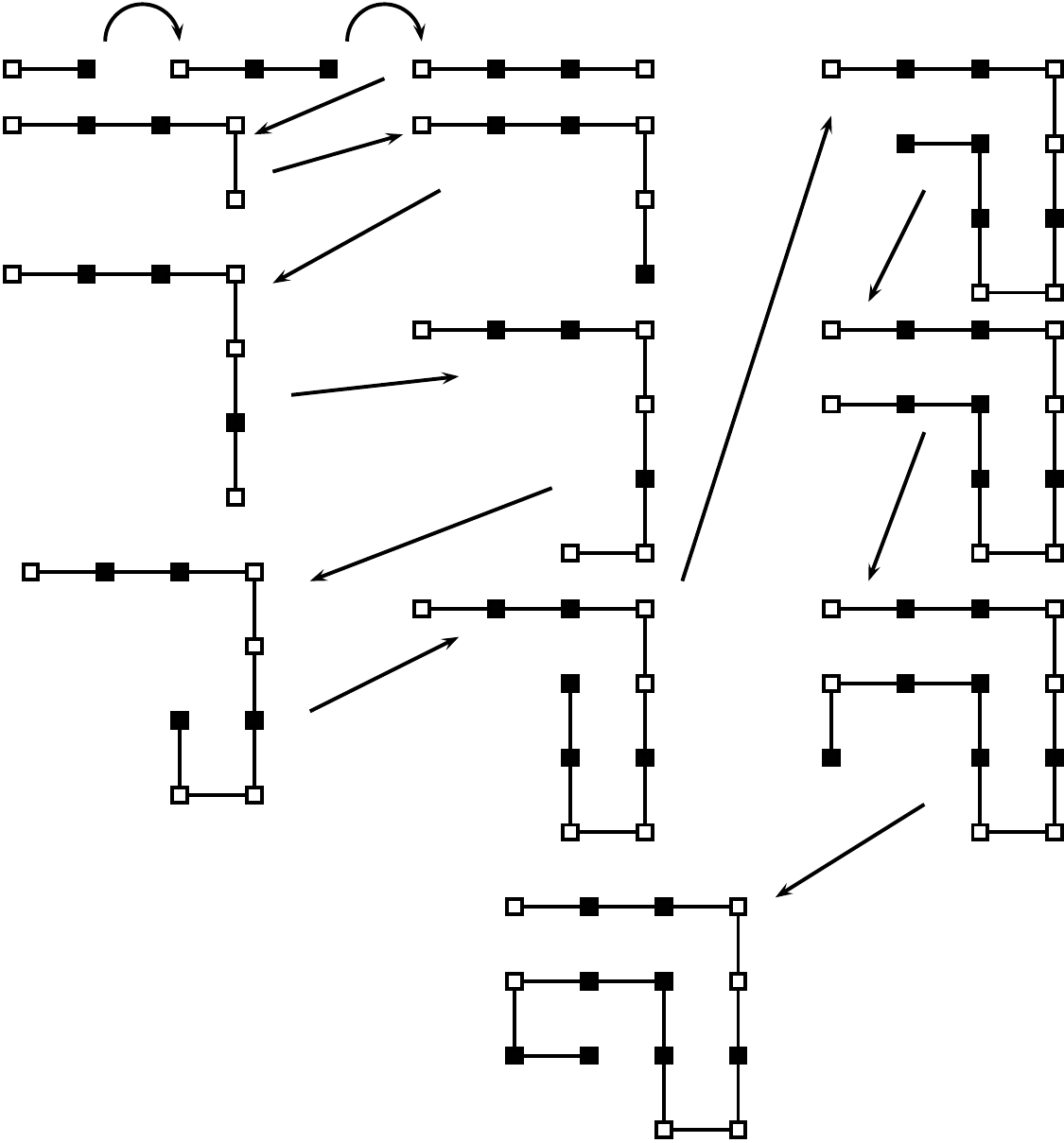}
\caption{Protein Structure Prediction by stretching SAWs}
\label{exElongation}
\end{figure}

In the remainder of this document, we give a more rigorous definition
of the $fSAW_n$ set, we initiate its study, and compare it to the
well-known $\mathcal{S}_n$ SAWs set.

\subsection{Notations}

Folded self-avoiding walks can be studied in a lattice having $d$ dimensions.
However, for the sake of simplicity, authors of this research work have
decided to introduce them only on the 2 dimensional square lattice $\mathds{Z}^2$,
to be as closed as possible to their field of application: the low resolution
backbone structure prediction of a protein. Such restriction enables us to 
produce understandable pictures of such not yet investigated particular walks.

One of the easiest way to define the folded self-avoiding walks
described previously, that appear during the realization of the SAW
requirement in PSP algorithms, is to introduce the absolute encoding
of a walk~\cite{Hoque09,Backofen99algorithmicapproach}.  In this
encoding, a $n+1-$step walk $w=w(0), \ldots, w(n)\in
\left(\mathds{Z}^2\right)^{n+1}$ with $w(0)=(0,0)$ is a sequence
$s=s(0), \ldots, s(n-1)$ of elements belonging into
$\mathds{Z}/4\mathds{Z}$, such that:
\begin{itemize}
\item $s(i)=0$ if and only if $w(i+1)_1=w(i)_1+1$ and $w(i+1)_2=w(i)_2$, that is, $w(i+1)$ is at the East of $w(i)$.
\item $s(i)=1$ if and only if $w(i+1)_1=w(i)_1$ and $w(i+1)_2=w(i)_2-1$: $w(i+1)$ is at the South of $w(i)$.
\item $s(i)=2$ if and only if $w(i+1)_1=w(i)_1-1$ and $w(i+1)_2=w(i)_2$, meaning that $w(i+1)$ is at the West of $w(i)$.
\item Finally, $s(i)=3$ if and only if $w(i+1)_1=w(i)_1$ and $w(i+1)_2=w(i)_2+1$ ($w(i+1)$ is at the North of $w(i)$).
\end{itemize}

Let us now define the following functions~\cite{guyeux:hal-00795127}.

\begin{definition}
  The \emph{anticlockwise fold function} is the function $f:
  \mathds{Z}/4\mathds{Z} \longrightarrow \mathds{Z}/4\mathds{Z}$
  defined by $f(x) = x-1 ~(\textrm{mod}~ 4)$ and the clockwise fold
  function is $f^{-1}(x) = x+1 ~(\textrm{mod}~ 4)$.
\end{definition}

Using the absolute encoding sequence $s$ of a $n-$step SAW $w$ that
starts from the origin of the square lattice, a pivot move of $+90$°
on $w(k)$, $k<n$, simply consists to transform $s$ into
$s(0),\hdots,s(k-1),f(s(k)),\hdots,f(s(n))$. Similarly, a pivot move
of $-90$° consists to apply $f^{-1}$ to the queue of the absolute
encoding sequence, like in Figure~\ref{ex1}.

\subsection{A graph structure for SAWs folding process}
\label{definitions}
We can now introduce a graph structure describing well the iterations of
$\pm 90$° pivot moves on a given self-avoiding walk.

Given $n \in \mathds{N}^*$, the graph $\mathfrak{G}_n$, formerly
introduced in~\cite{guyeux:hal-00795127}, is defined as follows:
\begin{itemize}
\item its vertices are the $n-$step self-avoiding walks, described in absolute 
encoding;
\item there is an edge between two vertices $s_i$, $s_j$ if and only
  if $s_j$ can be obtained by one pivot move of $\pm 90$° on $s_i$,
  that is, if there exists $k \in \llbracket 0, n-1\rrbracket$ s.t.:
\begin{itemize}
\item either $s_i(0),\hdots,s_i(k-1),f(s_i(k)),\hdots,f(s_i(n)) = s_j$
\item or $s_i(0),\hdots,s_i(k-1),f^{-1}(s_i(k)),\hdots,f^{-1}(s_i(n)) = s_j$.
\end{itemize}
\end{itemize}
Such a digraph is depicted in Figure~\ref{digraph}. The circled vertex is the
straight line whereas strikeout vertices are walks that are not self-avoiding.
Depending on the context, and for the sake of simplicity, $\mathfrak{G}_n$ will 
also refer to the set of SAWs in $\mathfrak{G}_n$ (\emph{i.e.}, its vertices).

\begin{figure}
\centering
\subfigure[000111]{\includegraphics[scale=0.5]{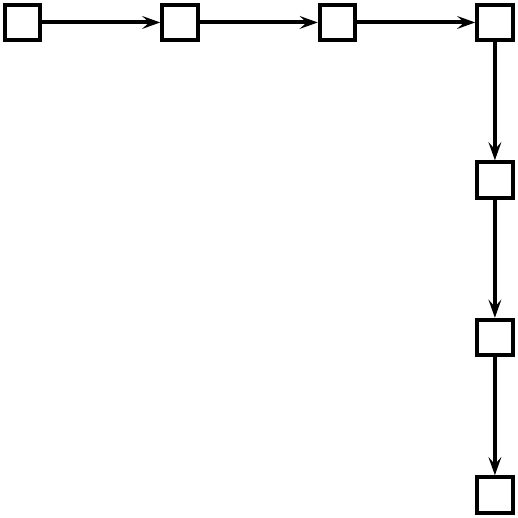}}\hspace{2cm}\subfigure[$001222=00f^{-1}(0)f^{-1}(1)f^{-1}(1)f^{-1}(1)$]{\includegraphics[scale=0.5]{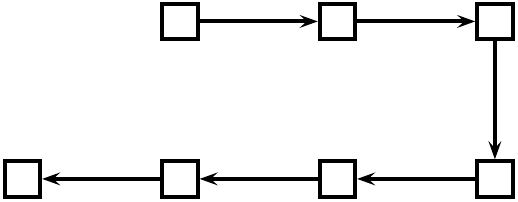}}
\caption{Effects of the clockwise fold function applied on the four
  last components of an absolute encoding.}
\label{ex1}
\end{figure}

Using this graph, the folded SAWs introduced in the previous section
can be redefined more rigorously.
\begin{definition}
  $fSAW_n$ is the connected component of the straight line $00\hdots
  0$ ($n$ times) in $\mathfrak{G}_n$, whereas $\mathcal{S}_n$ is
  constituted by all the vertices of $\mathfrak{G}_n$.
\end{definition}

The Figure~\ref{SokalnrSAW} shows that the connected component
$fSAW(223)$ of the straight line in $\mathfrak{G}_{223}$ is not equal
to the whole graph: $\mathfrak{G}_{223}$ is not connected. More
precisely, this graph has a connected component of size 1: Figure~\ref{SokalnrSAW} is 
totally unfoldable, whereas SAW of Fig.~\ref{saw3pasSaw4} can be folded 
exactly once. Indeed, to be in the same connected component is an
equivalence relation $\mathcal{R}_n$ on $\mathfrak{G}_n, \forall n \in
\mathds{N}^*$, and two SAWs $w$, $w'$ are considered equivalent (with
respect to this equivalence relation) if and only if there is a way to
fold $w$ into $w'$ such that all the intermediate walks are
self-avoiding. When existing, such a way is not necessarily unique.

These remarks lead to the following definitions.

\begin{definition}
Let $n\in\mathds{N}^*$ and $w \in \mathcal{S}_n$. We say that:
\begin{itemize}
\item $w$ is \emph{unfoldable} if its equivalence class, with
  respect to $\mathcal{R}_n$, is of size 1;
\item $w$ \emph{is a folded self-avoiding walk} if its equivalence
  class contains the $n-$step straight walk $000\hdots 0$ ($n-1$
  times);
\item $w$ \emph{can be folded $k$ times} if  a simple path of
  length $k$ exists between $w$ and another vertex in the same connected
  component of $w$.
\end{itemize}
Moreover, we introduce the following sets:
\begin{itemize}
\item $fSAW(n)$ is the equivalence class of the $n-$step straight walk, or the set of all folded SAWs.
\item $fSAW(n,k)$ is the set of equivalence classes of size $k$ in $(\mathfrak{G}_n,\mathcal{R}_n)$.
\item $USAW(n)$ is the set of equivalence classes of size 1
  $(\mathfrak{G}_n,\mathcal{R}_n)$, that is, the set of unfoldable walks.
\item $f^1SAW(n)$ is the complement of $USAW(n)$ in
  $\mathfrak{G}_n$. This is the set of SAWs on which we can apply at
  least one pivot move of $\pm 90°$.
\end{itemize}
\end{definition}

\begin{example}
  Figure~\ref{fig:2self-avoiding-219.2} shows the two elements of a
  class belonging into $fSAW(219,2)$ whereas Fig.~\ref{SokalnrSAW} is
  an element of $USAW(223)$.
\end{example}


\begin{figure}
\includegraphics[scale=0.45]{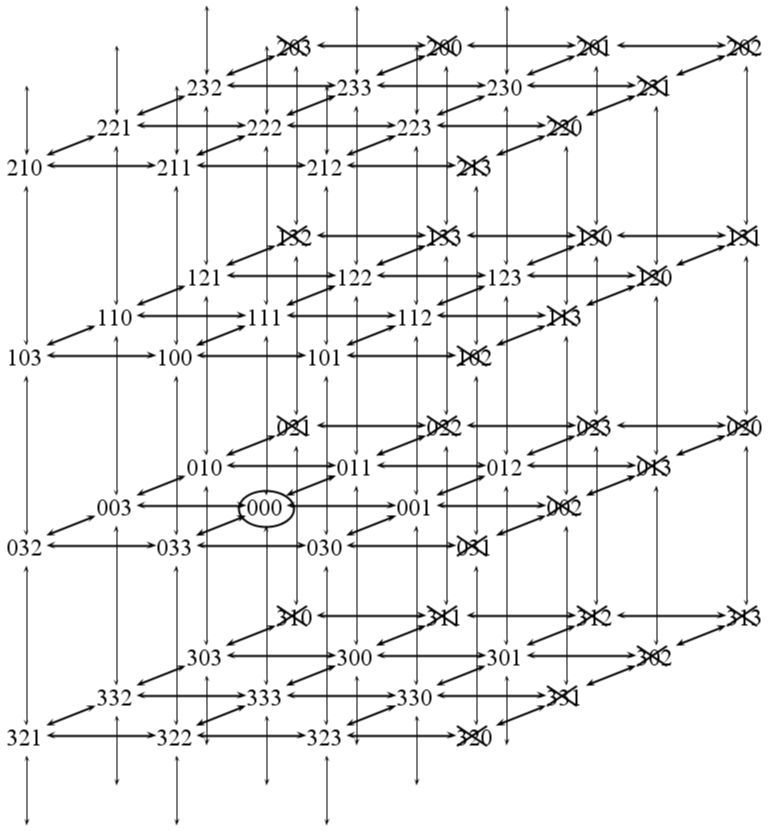}
\caption{The digraph $\mathfrak{G}_3 = fSAW(3)$}
\label{digraph}
\end{figure}

\begin{figure}
  \centering
  \subfigure{\includegraphics[scale=0.4]{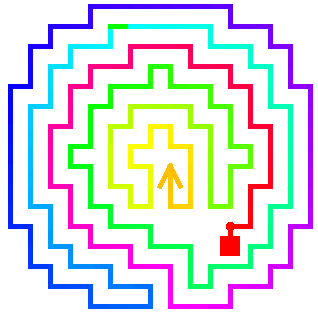}} 
  \hspace{2cm} 
  \subfigure{\includegraphics[scale=0.4]{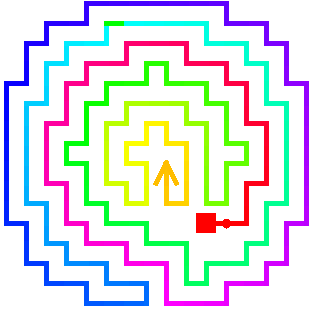}} 
  \caption{The two self-avoiding walks in
    $fSAW(219,2)$\label{fig:2self-avoiding-219.2}}
\end{figure}

\section{A Short List of Results on (un)folded Self-Avoiding Walks}
\label{sec:automatic gen}

We now give a first collection of easy-to-obtained results concerning
the particular SAW sets introduced in the previous section. These results
have been either obtained mathematically or by using computers. 

We firstly show that,

\begin{proposition}
\label{prop3}
The cardinality $\phi_n$ of $fSAW_n$ satisfies: $2^{n+2}\leqslant \phi_n \leqslant 4\times 3^n$. 
\end{proposition}

This result is a consequence of the following lemma.
\begin{lemma}
\label{lemme}
The $2^n$ $n-$step walks that take steps only in the positive
coordinate directions are in $fSAW(n)$.
\end{lemma}

This lemma can be proven using the number of cranks of a self-avoiding walk, 
defined below.

\begin{definition}[Crank]
Let $w$ be a $n-$step self-avoiding walk on $\mathds{Z}^2$ of absolute encoding
$s$. $w$ contains 
a crank at position $k\in\llbracket 1, n\rrbracket$ if $s(k-1)\neq s(k)$.
\end{definition}

\begin{proof}[Lemma~\ref{lemme}]
  Let $n \in \mathds{N}^*$.  We  show by a mathematical induction
  that, $\forall N \in \mathds{N}$, any $n-$step self-avoiding walk
  that (1) takes steps only in the positive coordinate directions, and
  (2) has $N$ cranks, is in $fSAW(n)$.

  The base case is obvious, as if $N=0$, then $w$ is a straight line.

  Let $N \in \mathds{N}$ such that the statement holds for all
  $k\leqslant N$, and consider a $n-$step self-avoiding walk $w$ that
  has $N+1$ cranks while taking steps only in the positive coordinate
  directions.  Let $j$ be the position of the first crank in $w$.  As
  steps are taken only in the positive coordinate directions, only two
  situations can occur (see Figure~\ref{preuve}):
  (1) $w(j)=w(j-1)+(1,0)$ and $w(j+1)=w(j)+(0,1)$ ($s(j-1)=0, s(j)=3$), or
  (2) $w(j)=w(j-1)+(0,1)$ and $w(j+1)=w(j)+(1,0)$ ($s(j-1)=3, s(j)=0$).

  Suppose now that the origin of the 2D square lattice is set to
  $w(j)$. So, in the first situation (1),

\begin{itemize}
  \item $\forall l>j$, $w(l)=(w(l)_1,w(l)_2)$ is such that $w(l)_1\geqslant 0$ while $w(l)_2\geqslant 1$,
  \item $\forall l<j,$ $w(l)=(w(l)_1,w(l)_2)$ is such that $w(l)_1\leqslant -1$ while $w(l)_2\leqslant 0$.
\end{itemize}
  The effect of a $90°$ pivot move on the origin $w(j)$ is to reduce
  the number of cranks $N+1$ to $N$ in $w$, and to map each $w(l)=(w(l)_1,w(l)_2)$
  into $(w(l)_2,w(l)_1)$, $\forall l>j$. After such a pivot move, the obtained
  walk $w'$ is such that $\forall l>j$, $w'(l)_1=w(l)_2\geqslant 1$,
  while $\forall l<j$, $w'(l)_1=w(l)_1\leqslant -1$.  In other words,
  the walk $w'$ still remains self-avoiding.  $w'$ having $N$ cranks,
  it belongs into $fSAW(n)$ due to the induction hypothesis.
  Furthermore, $w'$ is obtained by operating a pivot move on $w$, thus
  these two walks belong into the same connective component of
  $\mathfrak{G}_n$. Finally, $w \in fSAW(n)$.

  The second situation (2) also can be handled in that way, which concludes
  the mathematical induction and the proof of the lemma.
\end{proof}

\begin{proof}[Proposition~\ref{prop3}]
  Due to Lemma~\ref{lemme}, we have $\phi_n \geqslant 4\times 2^{n}$
  ($4 \times$ because of the 4 quarters of the square lattice).  And
  since the set of $n-$step walks without immediate reversals has
  cardinality $4\times 3^n$ and contains all $n-$step folded
  self-avoiding walks, we have $\phi_n \leqslant 4\times 3^n$.
\end{proof}

\begin{figure}
\centering
\includegraphics[scale=0.4]{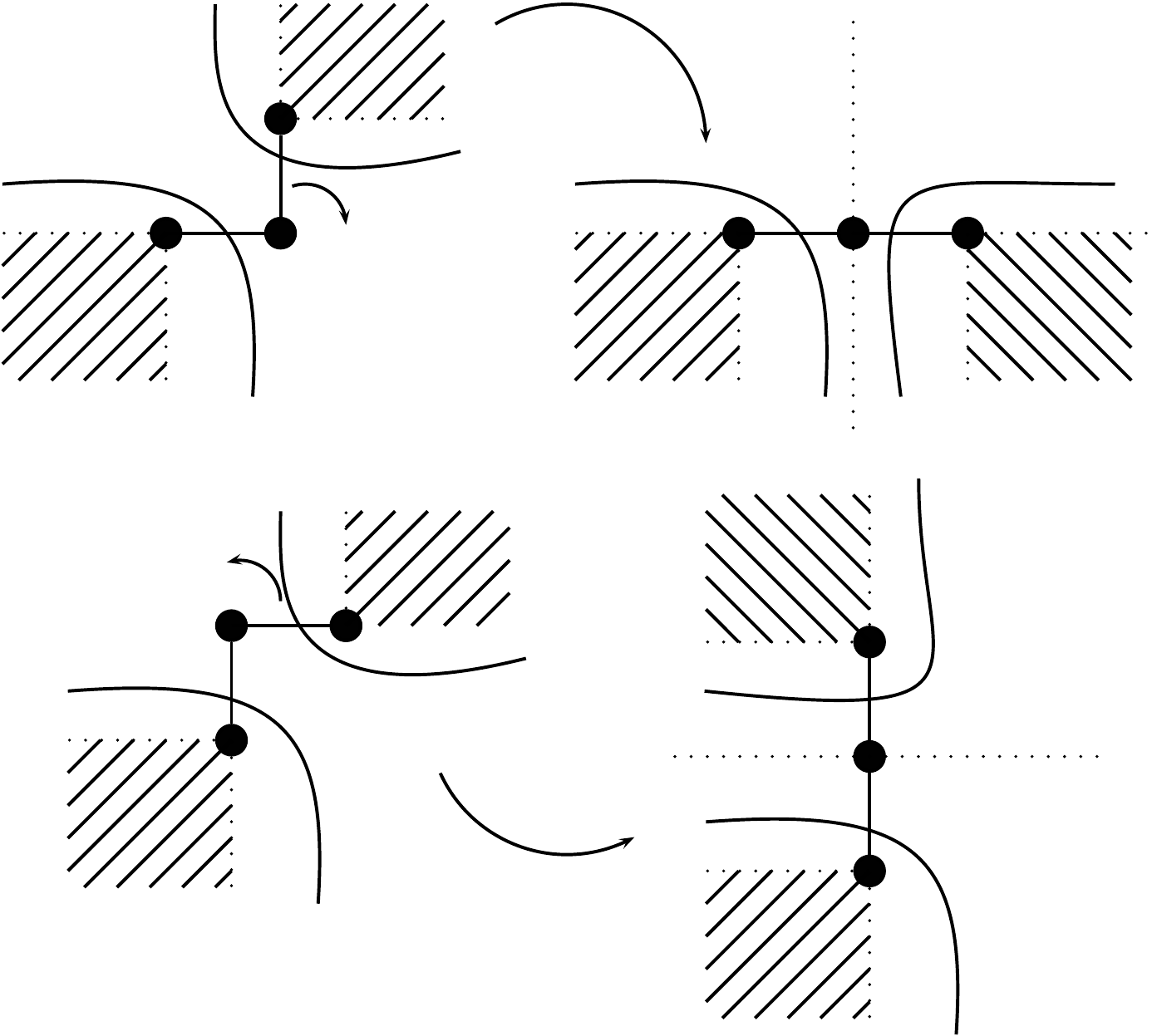}
\caption{Walks that contain only 3 and 0 in their absolute encoding
  are folded SAWs: reducing the number of cranks does not introduce
  intersections in the walk.}
\label{preuve}
\end{figure}

\begin{remark}
In particular, SAWs whose absolute encoding is only constituted by 0's and 1's
are folded SAWs. It is quite possible that a few 2's or 3's can be added without
breaking the folded character of the walk, meaning that the lower bound could be
increased.
\end{remark}

We can now give a result regarding the $USAW(n)$ set of self-avoiding walks.

\begin{theorem}
There is an infinite number of $n$ such that $USAW(n)$ is nonempty.
In particular, the number of unfoldable SAWs is infinite.
\end{theorem}

\begin{proof}
A proof of this result, too long to be contained in this work, can be
found in~\cite{articleTheoreme}. It consists to create a recursive construction
process of unfoldable self-avoiding walks, as depicted in Figure \ref{rtuc}.
\end{proof}

\begin{figure}[h!]
\begin{center}
  \subfigure[$w_0$ (239-step walk)]{\includegraphics[scale=0.2]{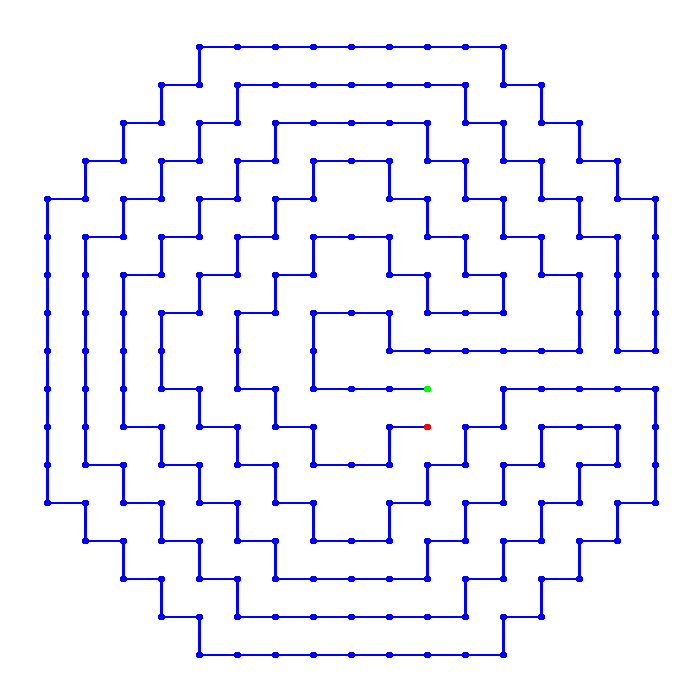}\label{nfSAW0}}\quad 
  \subfigure[$w_1$ (391-step walk)]{\includegraphics[scale=0.2]{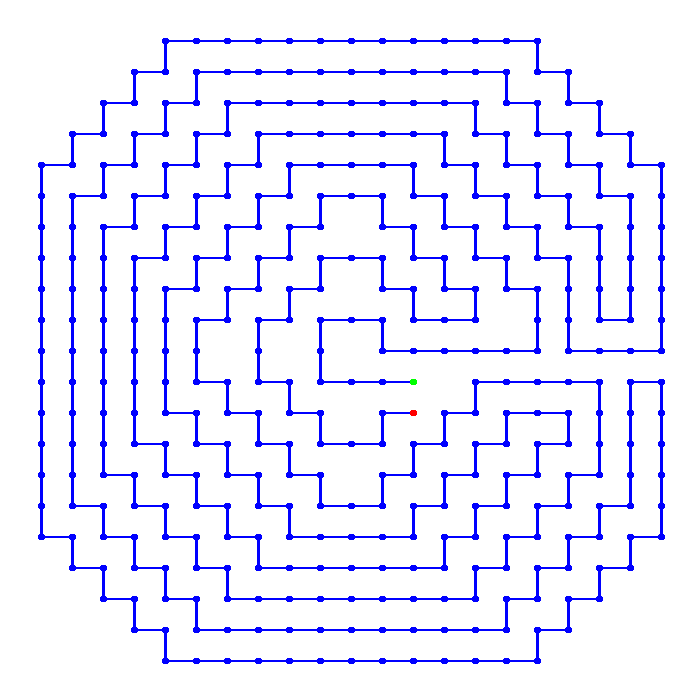}}\\
  \subfigure[$w_2$ (575-step walk)]{\includegraphics[scale=0.2]{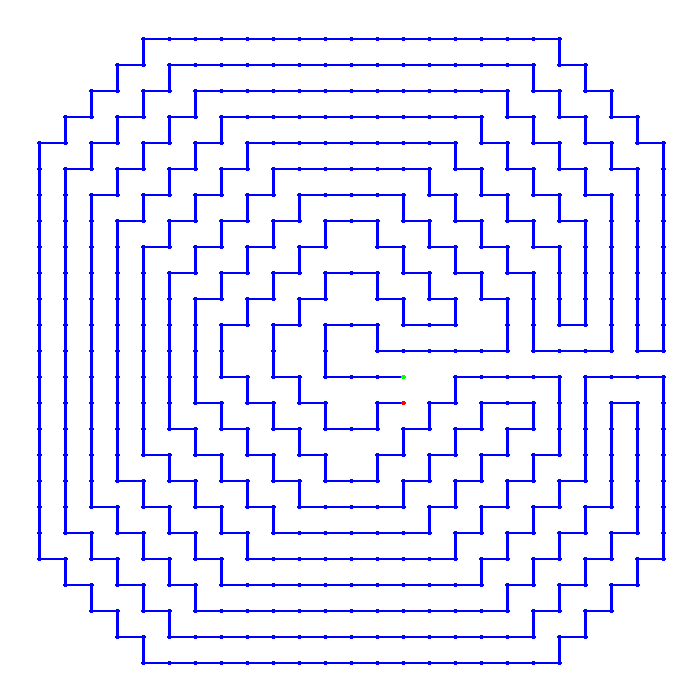}}\quad 
  \subfigure[$w_3$ (791-step walk)]{\includegraphics[scale=0.2]{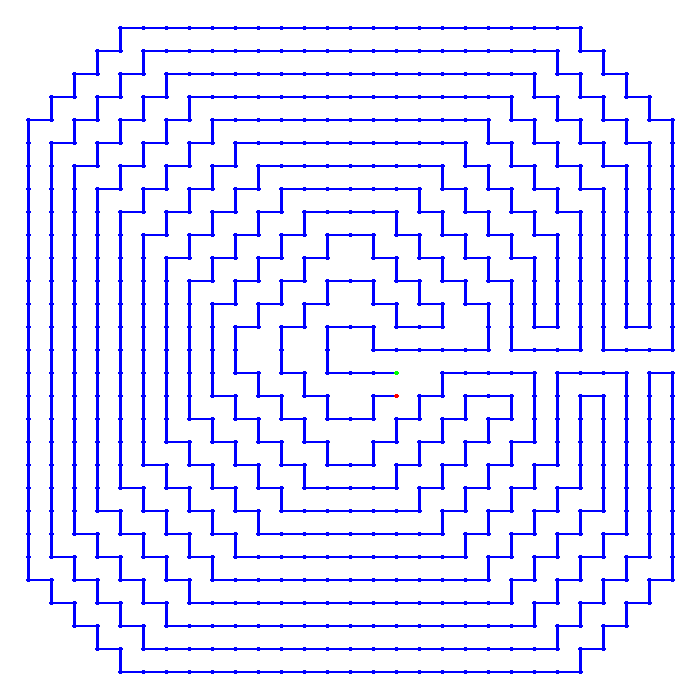}}
\end{center}
\caption{Generating walks that cannot be folded out}
\label{rtuc}
\end{figure}

\begin{proposition}
 $\forall n \leqslant 14, fSAW(n) = \mathfrak{G}_n$ whereas
$fSAW(107) \subsetneq \mathfrak{G}_{107}$ (see Figure~\ref{Vien}).

\noindent In other words, let $\nu_n$ the smallest $n\geqslant 2$ such
that $USAW(n) \neq \emptyset$.  Then $15\leqslant \nu_n\leqslant
107$.
\end{proposition}

\begin{figure}
 \begin{center}
\subfigure[$\mathfrak{G}_n$ for $n\leqslant 14$]{\includegraphics[scale=0.7]{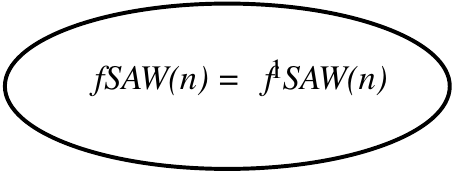}}\quad
\subfigure[Diagram of $\mathfrak{G}_n$ for $n=107$]{\includegraphics[scale=0.7]{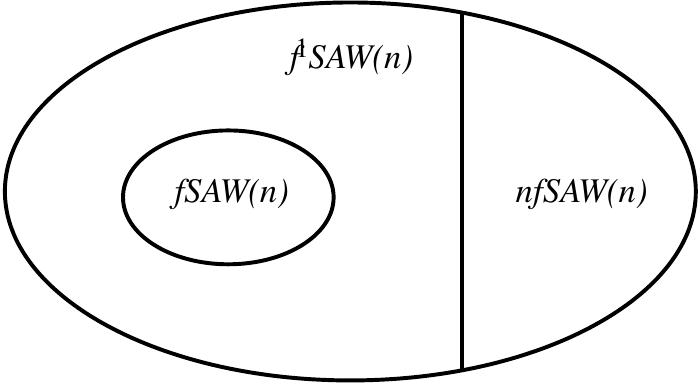}}\\
\caption{Vien diagram for $\mathfrak{G}_n$}
\label{Vien}
\end{center}
\end{figure}

\begin{proof}
  We have computed a program that constructs the connected component of
  the $n-$step straight line for $n\leqslant 14$, and at each time, we
  have obtained the whole $\mathfrak{G}_n$
  (see~\cite{articleCalculs}).  Additionally, we have obtained using a
  backtracking method the walk depicted in Figure~\ref{saw107}, which
  justifies the upper bound of 107: we have verified using a
  systematic program that no pivot move can be realized in that walk
  without breaking the self-avoiding requirement.  These programs,
  their explanations and justifications can be found
  in~\cite{articleCalculs}.
\end{proof}

\begin{figure}
\centering
\includegraphics[scale=0.5]{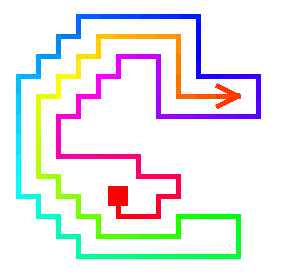}
\caption{Current smallest (107-step) SAW that cannot be folded out}
\label{saw107}
\end{figure}

\begin{proposition}
$\forall n \leqslant 28, f^1SAW(n) = \mathfrak{G}_n$.
\end{proposition}

\begin{proof}
  Obtained experimentally, see~\cite{articleCalculs}.

  The results contained into the two previous propositions are
  summarized, with all intermediate computations, in
  Table~\ref{composante connexe}.  The $\sharp\mathfrak{G}_n$ values,
  obtained in~\cite{Jensen04a}, are recalled here for comparison.
\end{proof}

\begin{table}
 \centering
 \begin{tabular}{c|c|c|c|c}
 $n$ & $\sharp\mathfrak{G}_n$ & $\sharp f^1SAW(n)$ & $\sharp USAW(n) = \sharp \overline{f^1SAW(n)}$ & $\sharp fSAW(n)$\\
\hline
1 & 4 & 4 & 0 & 4 \\
2 & 12 & 12 & 0 & 12 \\
3 & 36 & 36  & 0 &  36 \\
4 & 100 & 100 & 0 & 100 \\
5 & 284 & 284 & 0 & 284 \\
6 & 780 & 780  & 0 & 780\\
7 & 2172 & 2172 & 0 & 2172 \\
8 & 5916 & 5916  & 0 & 5916 \\
9 & 16268 & 16268 & 0 & 16268  \\
10 & 44100 & 44100 & 0 & 44100 \\
11 & 120292 & 120292 & 0 & 120292 \\
12 & 324932 & 324932 & 0 & 324932  \\
13 & 881500 & 881500 & 0 & 881500 \\
14 & 2374444 & 2374444 & 0 & 2374444 \\
15 & 6416596 & 6416596 & 0 & ? \\
16 & 17245332 & 17245332 & 0 & ? \\
17 & 46466676 & 46466676 & 0 & ?\\
18 & 124658732 & 124658732 & 0 & ?\\
19 & 335116620 & 335116620 & 0 & ?\\
20 & 897697164  & 897697164 & 0 & ?\\
21 & 2408806028  & 2408806028 & 0 & ?\\
22 & 6444560484  & 6444560484 & 0 & ?\\
23 & 17266613812  & 17266613812 & 0 & ?\\
24 & 46146397316  & 46146397316 & 0 & ?\\
25 & 123481354908  & 123481354908 & 0 & ?\\
26 & 329712786220  & 329712786220 & 0 & ?\\
27 & 881317491628  & 881317491628 & 0 & ?\\
28 & 2351378582244  & 2351378582244 & 0 & ?\\ 
29 & 6279396229332   & ? & ? & ?\\
30 & 16741957935348   & ? & ? & ?\\
31 & 44673816630956   & ? & ? & ? \\
\vdots & \vdots & \vdots & \vdots & \vdots \\
107 & ? & ? & $\geqslant 1$ & ? \\
 \end{tabular}
\caption{Cardinalities of various subsets of SAWs}
\label{composante connexe}
\end{table}

Until now, connected components presented in this paper either have
the straight line, or are of size 1 or 2. A reasonable questioning is
to wonder whether it is possible to have larger connected components
different from the one of the straight line. We are founded to claim
that,

\begin{proposition}
It exists $k>2$ such that $fSAW(n,k)$ is nonempty.
\end{proposition}

In other words, connected components different from $fSAW(n)$ and 
larger than 1 or 2 elements exist. 
The result, which has been experimentally obtained,
can be proven by exhibiting a counterexample:
Figure~\ref{sawConnected} shows
a connected component of size 5.

\begin{figure}
  \begin{minipage}{0.3\textwidth}
    \centering\resizebox{\textwidth}{!}{
      \includegraphics[scale=0.5]{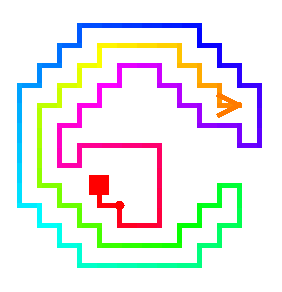}
    }
  \end{minipage}
  \begin{minipage}{0.3\textwidth}
    \centering\resizebox{\textwidth}{!}{
      \includegraphics[scale=0.5]{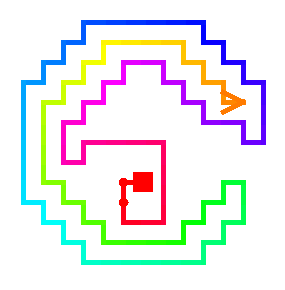}
    }
  \end{minipage}
  \begin{minipage}{0.3\textwidth}
    \centering\resizebox{\textwidth}{!}{
      \includegraphics[scale=0.5]{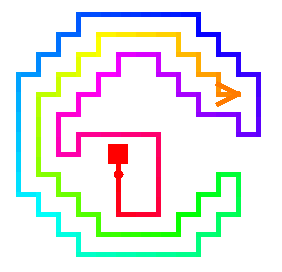}
    }
  \end{minipage}
  \begin{minipage}{0.3\textwidth}
    \centering\resizebox{\textwidth}{!}{
      \includegraphics[scale=0.5]{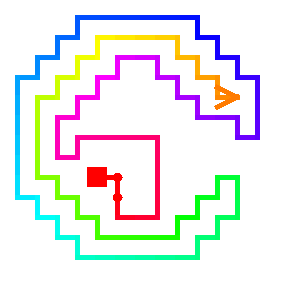}
    }
  \end{minipage}
  \begin{minipage}{0.3\textwidth}
    \centering\resizebox{\textwidth}{!}{
      \includegraphics[scale=0.5]{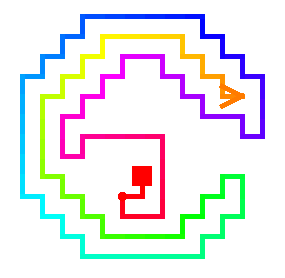}
    }
  \end{minipage}
  \caption{A connected component with 5 elements}
  \label{sawConnected}
\end{figure}

We can define a diameter function $D$ on the connected components of
$\mathfrak{G}_n$, such that $D(C)$ is the length of the longest
shortest path in the connected component $C$ of $\mathfrak{G}_n$.
Consider the connected component of the straight line $fSAW(n)$, we
have the result,

\begin{proposition}
  The diameter of $fSAW(n)$ is equal to $2n$: $D(fSAW(n)) =
  2n$.
\end{proposition}

\begin{proof}
  We take the SAW $S_{z_1}$ defined as the zigzag $(0,1,0,1,0,...)$ and the
 $S_{z_2}$ defined as the zigzag $(2,1,2,1,2,...)$. 

 We can transform $S_{z_1}$ in $(2,3,2,3,2,...)$ by two pivot moves:
 $$(0,1,0,1...) \rightarrow (1,2,1,2,1,...) \rightarrow (2,3,2,3,2,...).$$ 
 Then two other pivot moves allow us to transform $(2,3,2,3,2,...)$ in
 $(2,1,0,1,0,...)$, that is,
 $$(2,3,2,3,2,...) \rightarrow (2,2,1,2,1,2,...)\rightarrow (2,1,0,1,0,1,...).$$
 As the respective visited vectices start by $(0,1), (1,2), (2,3),
 (2,2), (2,1)$, we obtain by doing so a simple path of length 4.  The
 process can be reproduced on the queue $(0,1,0...)$ of $
 (2,1,0,1,0...)$ until each 0's (odd positions) of the SAW has been
 transformed to 2, and each 1's (even position) has been set again to
 1. As there are two pivot moves for each value in the path and each
 pivot moves is in a different direction in $\mathfrak{G}_n$, so 
 the minimum distance from $S_{z_1}$ to $S_{z_2}$ in $\mathfrak{G}_n$
 is $2n$.

 This path, from $S_{z_1}$ to $S_{z_2}$, is the largest distance we
 can find in $\mathfrak{G}_n$ as we have two pivot moves on each
 edge. If we add indeed one more pivot move, i.e., three pivot moves,
 on an edge then the same value could be obtained from the initial
 position by making only one pivot move in the opposite direction
 which would reduce the distance between the two SAWs.
\end{proof}

\begin{example}
  In $fSAW(2)$, this diameter corresponds, for instance, to the
  shortest path $03 \rightarrow 00 \rightarrow 11 \rightarrow 12
  \rightarrow 23$ (see Figure~\ref{digraph2}).
\end{example}

\begin{figure}
\centering
\includegraphics[scale=0.5]{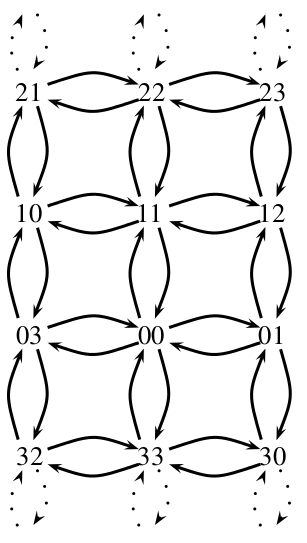}
\caption{The digraph $\mathfrak{G}_2 = fSAW(2)$}
\label{digraph2}
\end{figure}

\section{A list of Open Questions}
\label{sec:openquestions}

We enumerate in this section a list of open questions that have
appeared to us as interesting. Some of them should be very easy to solve,
whereas other ones may involve a degree of difficulty.

In the following we define $fSAW^d(n)$ as the class of equivalency of
the $n-$step straight walk on $\mathds{Z}^d$ and $\mathfrak{G}_n^d$ is
the equivalent of $\mathfrak{G}_n$ in $\mathds{Z}^d$. Note that $fSAW^2(n)$
is equal to $fSAW(n)$, as introduced in Definition~\ref{definitions}.

\begin{enumerate}
\item For any dimension $d$, do we have the existence of $n \in
  \mathds{N}^*$ such that $fSAW^d(n) \subsetneq \mathfrak{G}_n^d$?

\item $fSAW^2(2)$ and $fSAW^2(3)$ are obviously connected graphs, but
  they are not Eulerian. Indeed, more than two vertices have an odd
  degree both in $fSAW^2(2)$ and $fSAW^2(3)$ (see
  Figures~\ref{digraph2} and~\ref{digraph}). Is it the case for all
  $fSAW^d(n)$ ?

\item  $fSAW^2(2)$ and $fSAW^2(3)$ are Hamiltonian graphs, with the following 
Hamiltonian circuits:
\begin{itemize}
\item $00 \rightarrow 03 \rightarrow 32 \rightarrow 23 \rightarrow 10 \rightarrow 11 \rightarrow 22 \rightarrow 33 \rightarrow 30 \rightarrow 21 \rightarrow 12 \rightarrow 01 \rightarrow 00$ for $fSAW^2(2)$ (see Figure~\ref{digraph2}). 
\item $000 \rightarrow 003 \rightarrow 010 \rightarrow 011 \rightarrow 012 \rightarrow 001 \rightarrow 030 \rightarrow 323 \rightarrow 330 \rightarrow 301 \rightarrow 300 \rightarrow 333 \rightarrow 322 \rightarrow 321 \rightarrow 332 \rightarrow 303 \rightarrow 232 \rightarrow 233 \rightarrow 230 \rightarrow 223 \rightarrow 212 \rightarrow 211 \rightarrow 210 \rightarrow 221 \rightarrow 222 \rightarrow 111 \rightarrow 110 \rightarrow 121 \rightarrow 122 \rightarrow 123 \rightarrow 112 \rightarrow 101 \rightarrow 100 \rightarrow 103 \rightarrow 032 \rightarrow 033 \rightarrow 000$ for $fSAW^2(3)$ (see Figure~\ref{digraph}). 
\end{itemize}
Is it a coincidence, or is it the case for every $fSAW^d(n)$ ?

\item What is the exact value of the diameter $D(fSAW^d(n))$ ?

\item Do we have a connective constant for $fSAW^d(n)$. That is, does the limit
$\lim_{n \rightarrow +\infty}\phi_n^{1/n}$ exist, and can we bound it ?

\item $u_n=\sharp USAW^d(n)$ is an increasing sequence (for $d=2$, or
  for any $d$)?  Does it grow at a given (linear or exponential) rate?

\item Let $k \in \mathds{N}$. Is the sequence $v_n = \sharp fSAW(n,k)$
  increasing with $n$ ? If so, at which rate, and does it depend on the
  dimension $d$?  And what about the sequence $w_k = \sharp fSAW(n,k)$
  for a given $n$ ?

\item More simply, is there an unfoldable walk in
  $\mathds{Z}^3$ ?

\item Are the connected components of $\mathfrak{G}_n^d$ convex ? In
  other words, given two SAWs in a same component $C$. Are all (or at
  least one) the shortest paths connecting them on $\mathds{Z}^d$ 
  in $C$?

\item Is there a generating function expressing the folded
  self-avoiding walks more simply, making it possible to enumerate
  them on the square lattice (like what has been realized
  in~\cite{Conway1993}).

\item When we can fold a self-avoiding walk until a straight line,
  is it possible to fold it in such a way that the number of cranks
  decreases ? And for two given self-avoiding walks $w_i$ and $w_j$ of
  the same connected component of $\mathfrak{G}_n$, such that $w_i$
  has more cranks than $w_j$, is there a path from $w_i$ to $w_j$
  whose vertices' number of cranks is decreasing ? Is there a relation
  between the vertex depth and the number of cranks in $\mathds{Z}^d$?
\end{enumerate}

\section{Consequences on Protein Folding}
\label{sec:consequences}

This first theoretical study about folded self-avoiding walks raises several
questions regarding the protein structure prediction problem and the current
ways to solve it. In one category of PSP software, the protein is supposed to
be synthesized first as a straight line of amino acids, and then this line 
of a.a. is folded out until reaching a conformation that optimizes a given
scoring function. By doing so, the obtained backbone structures all belong
into $fSAW(n)$, where $n$ is the number of residues of the protein. The
second category of PSP software consider that, as the protein is already
in the aqueous solvent, it does not wait the end of the synthesis to take its
3D conformation. So they consider SAWs whose number of steps increases from
1 to the number of amino acids of the targeted protein and, at each step $k$,
the current walk is streched (one amino acid is added to the protein) in 
such a way that the pivot $k$ is placed in the position that optimizes the 
scoring function they consider. By doing so, the possible predicted backbones
are the whole $\mathfrak{G}_3$. The two sets of possible conformations are 
different, at least when considering 2D low resolution models.

\begin{figure}[h!]
\begin{center}
  \subfigure[Conformation having best score (27)]{\includegraphics[scale=0.3]{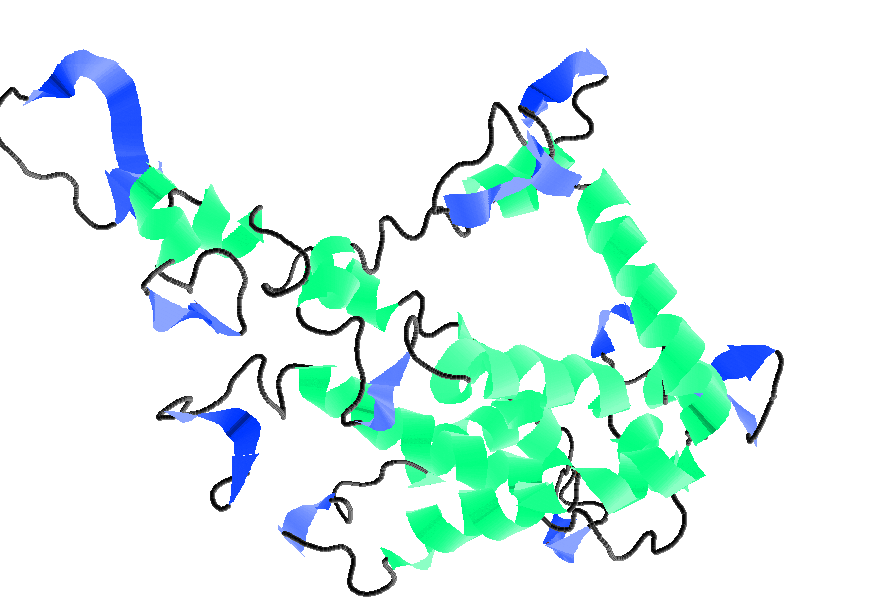}}\quad 
  \subfigure[Second best conformation (score 24)]{\includegraphics[scale=0.3]{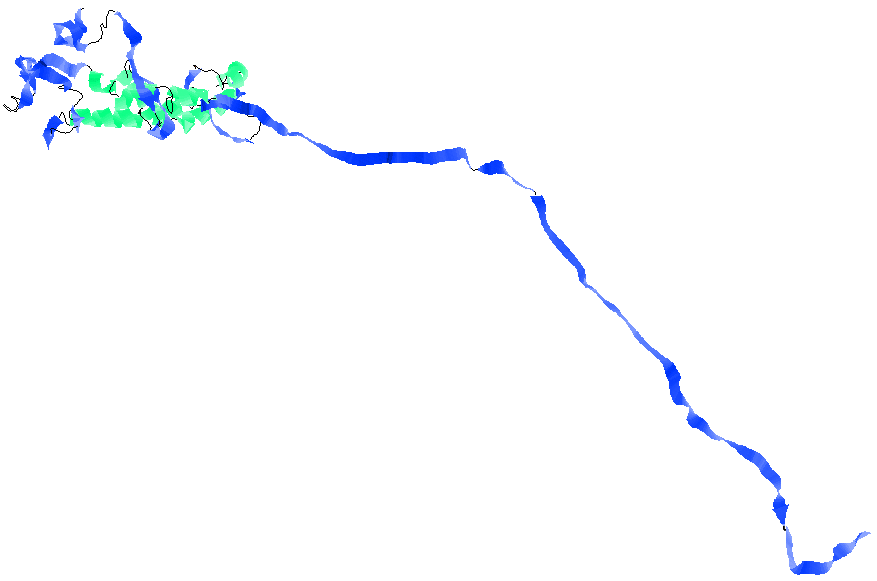}}\\
\end{center}
\caption{Illustration of chaos in protein folding (conformations have been predicted using RaptorX)}
\label{rtuc}
\end{figure}

We show by this work that (1) to take place in the first situation
(folding the straight line by a succession of pivot moves) can be
interesting as the number of possible SAW conformations is smaller
than $\sharp\mathfrak{G}_n$.  Indeed this interest is directly related
to the rate $\dfrac{\sharp fSAW(n)}{\sharp\mathfrak{G}_n}<1$. If this
rate decreases dramatically when $n$ increases, then the computational
advantage is obvious. However,we have currently no idea of such a
gain, that is, of the growing rate of $\sharp fSAW(n)$ compared to
$\sharp\mathfrak{G}_n<1$. (2) The use of heuristics instead of exact
methods (like SAT solvers for instance) is \emph{a priori} not
justified for PSP software that fold the straight line. Indeed, the
PSP problem has been proven NP hard on the set $\mathfrak{G}_n$ of all
possible SaWs. As they consider a strict subset of it, the complexity
of the problem might be reduced due to a lower number of cases to
consider. However, Proposition~\ref{prop3} tends to indicate that this
problem still remains difficult in $fSAW(n)$, which nevertheless
necessitates a rigorous complexity proof. (3) Biologically speaking,
to suppose that the proteins wait to be completely synthesized before
starting to fold appears as unrealistic, as the synthesis occurs in an
aqueous solvent. Indeed, the protein starts to fold during its
synthesis.  Furthermore, to the authors' opinion, it is restrictive to
consider that the head of the protein definitively stops to fold after
having synthesized. Such a supposition is equivalent to make a
confusion between local (the SAW at step $k$) and global (the final
optimal SAW) optimization. Indeed, authors of this manuscript
recognize honestly that they have no idea to determine if this third
approach (continuously folding the walk while stretching it) is more
reasonable than the previous ones, and if it is equivalent to either
$fSAW(n)$ or to $\mathfrak{G}_n$ (or if it constitutes a third
different subset of SAWs).

The authors' goal is only to point out the importance to determine the
best dynamical system to model protein folding before programming it
in PSP software, as this model determine which conformations can be
predicted.  A last remark to emphasize the importance of such a study:
authors of~\cite{bgc11:ip} have proven that the dynamical system used
in the ``folding the straight line'' category is chaotic according to
Devaney, meaning that any wrong choice of pivot move (due to
approximations in the scoring function, for instance) can potentially
become dramatic. Other researches (\cite{Braxenthaler97} for instance)
tend to show that the protein folding process intrinsically embeds a
certain amount of chaos. Thus, to use a more or less erroneous model
to predict the conformation could have grave consequences in
prediction quality.  Figure~\ref{rtuc} shows the two best
conformations predicted by RaptorX~\cite{raptorX}, a well-known PSP
software. We can see that using twice a same model, but with different
parameters can potentially lead to quite different conformations,
illustrating a possible effect of some chaotic properties exhibited by
the chosen model. We can reasonably wonder what is the effect of a
wrong model in such a prediction.

\section{Conclusion}

In this paper, the problem of self-avoiding walks folding in
the square lattice has been tackled. Regarding the protein structure 
prediction problem, we have shown
that the set of generated self-avoiding walks
depends on the PSP software category. In particular some particular
conformations cannot be reached by just folding the straight line
whereas they can be generated using random SAW generators as the pivot
algorithm. Starting from this fact, we have proposed a further exploration of
the folded self-avoiding walks. Different
subsets of  self-avoiding walks have been defined, 
like the set of unfoldable walks. We have shown that, even though their is an infinite number of
unfoldable SAWs, the number of folded SAWs is still exponential. After
having described the first obtained results on (un)folded SAWs,
we have proposed a list of open questions that could be explored on these
SAWs. Lastly, the link between (un)folded
SAWs and proteins has been questioned, and the consequences of the PSP software
choice on protein conformation has been highlighted.

Several research problems are interesting to further study and better
understand the properties of (un)folded SAWs, as shown in the open
questions section. Our future work will be concentrated on finding the
smallest unfolded SAWs, finding the smallest connected components of
unfolded SAWs, and on the optimization of energy levels of a given folded SAW.

\section{Acknowledgement}
The authors wish to thank Kamel Mazouzi, Thibaut Cholley, Raphaël
Couturier, and Alain Giorgetti for their help in understanding USAWs. All
the computations presented in the paper have been performed on the
supercomputer facilities of the Mésocentre de calcul de Franche-Comté.

\bibliographystyle{plain}
\bibliography{biblio}

\end{document}